\documentclass{llncs}
\usepackage[T1]{fontenc}
\usepackage{url}

\usepackage{stmaryrd}
\usepackage{amssymb}
\usepackage[all]{xy}

\newcommand{\setof}[1]{\{#1\}}
\newcommand{\setfunct}[1]{#1\ensuremath{_{\cal S}}}
\newcommand{\plural}[1]{#1\ensuremath{_{\cal P}}}

\newcommand{\pt}[1]{\ensuremath{\left\llbracket #1 \right\rrbracket _{\cal P}}} 
\newcommand{\ctxt}[1]{{\cal{C}}[#1]} 
\def\deriv#1{\buildrel #1 \over \to}
\newcommand{\tostar}{\mathrel{\deriv{*}}}

\renewcommand{\tt}{\ttfamily}
\newcommand{\codefont}{\small\tt}
\newcommand{\code}[1]{\mbox{\codefont{#1}}}
\newcommand{\ccode}[1]{``\code{#1}''}
\newcommand{\bs}{\char92} 
\newcommand{\us}{\raise-.8ex\hbox{-}}
\newcommand{\xtilde}{\!\raise-.75ex\hbox{\char`\~}} 

\newcommand{\equprogram}[1]{%
\def\separator{0.7ex}%
\frenchspacing%
\refstepcounter{equation}%
\par\vspace\separator\hspace{-0.5em}%
$\vcenter{\codefont\noindent{#1}}$%
\raisebox{-0.0ex}{\kern-1.20em\llap{\rm (\theequation)}}
\par\vspace\separator\noindent\kern-.0em%
}

\usepackage{listings}
\lstnewenvironment{curry}[1][]
  {\lstset{literate={->}{{$\rightarrow{}\!\!\!$}}3,#1}}{}
\newcommand{\listline}{\vrule width0pt depth1.5ex}

\newcommand{\COMMENT}[1]{}
\usepackage{color}

\begin{document}
\pagestyle{plain} 
\sloppy

\title{Synthesizing Set Functions}

\author{
Sergio Antoy\inst{1}
\kern1em
Michael Hanus\inst{2}
\kern1em
Finn Teegen\inst{2}
}
\institute{
Computer Science Dept., Portland State University, Oregon, U.S.A.\\
\email{antoy@cs.pdx.edu}\\[1ex]
\and
Institut f\"ur Informatik, CAU Kiel, D-24098 Kiel, Germany. \\
\email{\{mh,fte\}@informatik.uni-kiel.de}
}

\maketitle

\begin{abstract}
  Set functions are a feature of functional logic programming
  to encapsulate all results of a non-deterministic computation
  in a single data structure.
  Given a function $f$ of a functional logic program written in Curry,
  we describe a technique to synthesize the definition of the set
  function of $f$.  The definition produced by our technique is based
  on standard Curry constructs.
  Our approach is interesting for three reasons. It allows reasoning
  about set functions, it offers an implementation of set functions
  which can be added to any Curry system,
  and it has the potential of changing our thinking
  about the implementation of non-determinism,
  a notoriously difficult problem.
\end{abstract}

\section{Introduction}
\label{sec:introduction}

Functional logic languages, such as Curry and TOY,
combine the most important features of functional and logic languages.
In particular, the combination of lazy evaluation and
non-determinism leads to better evaluation strategies
compared to logic programming \cite{AntoyEchahedHanus00JACM}
and new design patterns \cite{AntoyHanus02FLOPS}.
However, the combination of these features poses new challenges.
In particular, the encapsulation of non-strict non-deterministic
computations has not a universally accepted solution so that
different Curry systems offer different implementations for it.
Encapsulating non-deterministic computations is an important
feature for application programming
when the task is to show whether some problem
has a solution or to compare different solutions in order to
compute the best one.
For this purpose, let $S(e)$ denote the set of all the values
of an expression $e$.
The problem with such an encapsulation operator is the fact
that the $e$ might share subexpressions which are defined outside $S(e)$.
For instance, consider the expression
\equprogram{\label{ex:outsidecapsule}%
let x = 0?1 in $S(\code{x})$
}
The infix operator \ccode{?} denotes a non-deterministic choice,
i.e., the expression \ccode{0?1} has \emph{two} values:
\code{0} or \code{1}.
Since the non-determinism of \code{x} is introduced outside $S(\code{x})$,
the question arises whether this should be encapsulated.
\emph{Strong encapsulation}, which is similar to Prolog's \code{findall},
requires to encapsulate all non-determinism
occurring during the evaluation of the encapsulated expression.
In this case, expression (\ref{ex:outsidecapsule}) evaluates
to the set $\{\code{0},\code{1}\}$.
As discussed in \cite{BrasselHanusHuch04JFLP},
a disadvantage of strong encapsulation is its dependence
on the evaluation strategy.
For instance, consider the expression
\equprogram{\label{ex:outsidecapsule-tuple}%
let x = 0?1 in ($S(\code{x})$, x)
}
If the tuple is evaluated from left to right,
the first component evaluates to $\{\code{0},\code{1}\}$
but the second component non-deterministically evaluates to the
values \code{0} and \code{1} so that the expression
(\ref{ex:outsidecapsule-tuple}) evaluates to the values
\code{($\{\code{0},\code{1}\}$,0)} and
\code{($\{\code{0},\code{1}\}$,1)}.
However, in a right-to-left evaluation of the tuple,
\code{x} is evaluated first to one of the
values \code{0} and \code{1} so that, due to sharing,
the expression (\ref{ex:outsidecapsule-tuple}) evaluates to the values
\code{($\{\code{0}\}$,0)} and \code{($\{\code{1}\}$,1)}.

To avoid this dependency on the evaluation strategy,
\emph{weak encapsulation} of $S(e)$ only encapsulates the
non-determinism of $e$ but not the non-determinism originating from
expressions created outside $e$.
Thus, \emph{weak encapsulation} produces the last result of
(\ref{ex:outsidecapsule-tuple}) independent of the evaluation
strategy.
Weak encapsulation has the disadvantage that its meaning
depends on the syntactic structure of expressions.
For instance, the expressions
\ccode{let x = 0?1 in $S(\code{x})$} and
``$S(\code{let x = 0?1 in x})$''
have different values.
To avoid misunderstandings and make the syntactic structure
of encapsulation explicit,
\emph{set functions} have been proposed \cite{AntoyHanus09}.
For any function $f$, there is a set function $\setfunct{f}$
which computes the set of all the values $f$ for given argument values.
The set function encapsulates the non-determinism caused
by the definition of $f$ but not non-determinism originating from arguments.
For instance, consider the operation
\equprogram{\label{ex:double}%
double x = x + x
}
The result of \code{\setfunct{double}} is always a set
with a single element since the definition of \code{double}
does not contain any non-determinism.
Thus, \code{\setfunct{double} (0?1)}
evaluates to the two sets $\{\code{0}\}$ and $\{\code{2}\}$.

Although set functions fit well into the framework of
functional logic programming,
their implementation is challenging.
For instance, the Curry system PAKCS \cite{Hanus17PAKCS}
compiles Curry programs into Prolog programs so that
non-determinism is implemented for free.
Set functions are implemented in PAKCS by Prolog's \code{findall}.
To obtain the correct separation of non-determinism caused by arguments
and functions, as discussed above, arguments are completely
evaluated before the \code{findall} encapsulation is invoked.
Although this works in many cases, there are some situations
where this implementation does not deliver any result.
For instance, if the complete evaluation of arguments fails
or does not terminate, no result is computed even if the set function
does not demand the complete argument values.
Furthermore, if the set is infinite, \code{findall} does not terminate
even if the goal is only testing whether the set is empty.
Thus, the PAKCS implementation of set functions is ``too strict.''

These problems are avoided by the implementation of set functions
in the Curry system KiCS2,
which compiles Curry programs into Haskell programs
and represents non-deterministic values in tree-like structures
\cite{BrasselHanusPeemoellerReck11}.
A similar but slightly different representation
is used to implement set functions.
Due to the interaction of different levels of non-determinism,
the detailed implementation is quite complex
so that a simpler implementation is desirable.

In this paper, we propose a new implementation of set functions
that can be added to any Curry system.
It avoids the disadvantages of existing implementations
by synthesizing an explicit definition of a set function
for a given function.
Depending on the source code of a function,
simple or more complex definitions of a set function are derived.
For instance, nested set functions require a more complex scheme
than top-level set functions, and functions with non-linear right-hand sides
require an explicit implementation of the call-time choice semantics.

The paper is structured as follows.
In the next section, we review some aspects of functional logic programming
and Curry.
After the definition of set functions in Sect.~\ref{sec:set-functions},
we introduce in Sect.~\ref{sec:plural} plural functions as
an intermediate step towards the synthesis of set functions.
A first and simple approach to synthesize set functions
is presented in Sect.~\ref{sec:toplevelsets}
before we discuss in Sect.~\ref{sec:nonlinear}
and \ref{sec:nested} the extension to non-linear rules and
nested set functions, respectively.
Sect.~\ref{sec:related} discussed related work before
we conclude in Sect.~\ref{sec:conclusions}.

\section{Functional Logic Programming and Curry}
\label{sec:flp}

We assume familiarity with the basic concepts of
functional logic programming \cite{AntoyHanus10CACM,Hanus13}
and Curry \cite{Hanus16Curry}.
Therefore, we briefly review only those aspects that are
relevant for this paper.

Although Curry has a syntax close to Haskell \cite{PeytonJones03Haskell},
there is an important difference in the interpretation of rules
defining an operation. If there are different rules that might
be applicable to reduce an expression, all rules
are applied in a non-deterministic manner.
Hence, operations might yield more than one result on a given input.
Non-deterministic operations, which are
interpreted as mappings from values into sets of values \cite{GonzalezEtAl99},
are an important feature of contemporary functional logic languages.
The archetype of non-deterministic operations is the
choice operator \ccode{?} defined by
\begin{curry}
x ? _ = x
_ ? y = y
\end{curry}
Typically, this operator is used to define other
non-deterministic operations like
\begin{curry}
coin = 0 ? 1
\end{curry}
Thus, the expression \code{coin} evaluates to one of the values \code{0} or
\code{1}.
Non-deterministic operations are quite expressive
since they can be used to completely eliminate logic variables
in functional logic programs, as show in
\cite{AntoyHanus06ICLP,deDiosCastroLopezFraguas07}.
Therefore, we ignore them in the formal development below.
For instance, a Boolean logic variable can be replaced
by the non-deterministic \emph{generator operation} for Booleans
defined by
\equprogram{\label{ex:abool}%
aBool = False ? True
}
Passing non-deterministic operations as arguments,
as in the expression \code{double$\;$coin},
might cause a semantical ambiguity.
If the argument \code{coin} is evaluated before calling \code{double},
the expression has two values, \code{0} and \code{2}.
However, if the argument \code{coin} is passed
unevaluated to the right-hand side of \code{double}
and, thus, duplicated, the expression has three different values:
\code{0}, \code{1}, or \code{2}.
These two interpretations are called \emph{call-time choice}
and \emph{run-time choice} \cite{Hussmann92}.
Contemporary functional logic languages stick to the call-time choice,
since this leads to results which are independent of the
evaluation strategy and has the rewriting logic CRWL \cite{GonzalezEtAl99}
as a logical foundation for declarative programming
with non-strict and non-deterministic operations.
Furthermore, it can be implemented by sharing which is already
available in implementations of non-strict languages.
In this paper, we use a simple reduction relation,
equivalent to CRWL,
that we sketch without giving all details (which can be found in
\cite{LopezRodriguezSanchez07}).

A \emph{value} is an expression without defined operations.
To cover non-strict computations, expressions can also
contain the special symbol $\bot$ to represent
\emph{undefined or unevaluated values}.
A \emph{partial value} is a value that might contain occurrences of $\bot$.
A \emph{partial constructor substitution} is a substitution
that replaces variables by partial values.
A \emph{context} $\ctxt{\cdot}$ is an expression with some ``hole.''
Then expressions are reduced according to the following
reduction relation:\\[1ex]
\begin{tabular}{@{\quad}r@{~~}c@{~~}l@{~~~}l}
 $\ctxt{f~\sigma(t_1) \ldots \sigma(t_n)}$ & $\to$ & $\ctxt{\sigma(r)}$ &
where  $f~t_1 \ldots t_n ~\code{=}~ r$ is a program rule\\
 & & & and $\sigma$ a partial constructor substitution\\[0ex]
 $\ctxt{e}$ & $\to$ & $\ctxt{\bot}$ &
where $e \ne \bot$
\end{tabular}\\[1ex]
The first rule models the call-time choice: if a rule is applied,
the actual arguments of the operation must have been evaluated to
partial values. The second rule models non-strictness by allowing
the evaluation of any subexpression to an undefined value
(which is intended if the value of this subexpression is not demanded).
As usual, $\tostar$ denotes the reflexive and transitive closure
of this reduction relation.
We also write $e = v$ instead of $e \tostar v$ if $v$ is a (partial) value.

For the sake of simplicity, we assume that programs are already translated
into a simple standard form:
conditional rules are replaced by if-then-else expressions
and the left-hand sides of all operations (except for \ccode{?})
are \emph{uniform} \cite{MorenoEtAl90ALP}
i.e., either the operation is defined by a single rule where all arguments
are distinct variables, or in the left-hand sides of all rules only
the last (or any other fixed) argument
is a constructor with variables as arguments
where each constructor of a data type occurs in exactly one rule.
In particular, rules are not overlapping so that
non-determinism is represented by \code{?}-expressions.

\section{Set Functions}
\label{sec:set-functions}

Set functions have been introduced in Sect.~\ref{sec:introduction}.
Their main feature is to encapsulate only the non-determinism
caused by the definition of the corresponding function.
Similarly to non-determinism, set functions encapsulate
failures only if they are caused by the function's definition.
If \code{failed} denotes a failing computation, i.e., an
expression without a value,
the expression \code{\setfunct{double}$\;$failed} has no value
(and not the empty set as a value).
Since the meaning of failures and nested set functions
has not been discussed in \cite{AntoyHanus09},
Christiansen et al.~\cite{ChristiansenHanusReckSeidel13PPDP} propose
a rigorous denotational semantics for set functions.
In order to handle failures and choices in nested applications
of set functions, computations are attached with
``nesting levels'' so that failures caused by different
encapsulation levels can be distinguished.

In the following, we use a simpler model of set functions.
Formally, set functions return sets.
However, for simplicity, our implementation returns multisets,
instead of sets, represented as some abstract data type.
Converting multisets into sets is straightforward
for the representation we choose.
If $b$ is a type, $\setof{b}$ denotes the type of a
set of elements of type $b$.
The meaning of a set function can be defined as follows:

\begin{definition}
\label{set-function-definition}
Given a unary (for simplicity) function $f :: a \to b$,
the \emph{set function} of $f$, $\setfunct{f} :: a \to \setof{b}$,
is defined as follows:
For every partial value $t$ of type $a$,
value $u$ of type $b$, and
set $U$ of elements of type $b$,
$(f\;t) = u$ iff $(\setfunct{f}\;t) = U$ and $u \in U$.
\end{definition}
This definition accommodates the different aspects of set functions
discussed in the introduction.
If the evaluation of an expression $e$ leads to a failure or choice
but its value is not required by the set function,
it does not influence the result of the set function
since $e$ can be derived to the partial value $\bot$.

\begin{example}\label{ex:ndconst}
Consider the following function:
\equprogram{%
  ndconst x y = x ? 1
}
The value of \code{\setfunct{ndconst}$\;$2$\;$failed} is $\{\code{2,1}\}$,
and \code{\setfunct{ndconst}$\;$(2?4)$\;$failed}
has the values $\{\code{2,1}\}$ and $\{\code{4,1}\}$.
\end{example}

Given a function $f$, we want to develop a method to synthesize
the definition $\setfunct{f}$.
A difficulty is that $f$ might be composed of other functions
whose non-determinism should also be encapsulated by $\setfunct{f}$.
This can be solved by introducing plural functions,
which are described next.

\section{Plural Functions}
\label{sec:plural}

If $f :: b \to c$ and $g :: a \to b$ are functions,
their composition $(f \circ g)$ is well defined
by $(f \circ g)(x) = f(g(x))$ for all $x$ of type $a$.
However, the corresponding set functions,
$\setfunct f :: b \to \setof c$ and $\setfunct g :: a \to \setof b$,
are not composable because their types mismatch---an
output of $\setfunct g$ cannot be an input of $\setfunct f$.
To support the composition of functions that return sets,
we need functions that take sets as arguments.\footnote{%
The notion of ``plural function'' is also used
in \cite{RiescoRodriguez14} to define a ``plural'' semantics
for functional logic programs. Although the type of their plural functions
is identical to ours, their semantics is quite different.}

\begin{definition}
  \label{plural-definition}
  Let $f :: a \to b$ be a function.
  We call \emph{plural} function of $f$ any
  function $\plural f :: \setof a \to \setof b$
  with the following property:
  for all $X$ and $Y$ such that $\plural f\;X = Y$,
  (1) if $y \in Y$ then there exists some $x \in X$ such that
  $f\;x = y$ and
  (2) if $x \in X$ and $f\;x = y$, then $y \in Y$.
\end{definition}
The above definition generalizes to functions with more than one
argument.  The following display shows both the type and an example of
application of the plural function, denoted by
\ccode{\plural{++}}, of the usual list concatenation \ccode{++}:
\equprogram{%
  (\plural{++}) :: \setof{[a]} -> \setof{[a]} -> \setof{[a]} \\
  \setof {[1],[2]} \plural{++} \setof {[],[3]} = \setof{[1],[1,3],[2],[2,3]}
}
Plural functions are unique, composable, and cover all the results
of set functions (see appendix for details).
Since plural functions are an important step towards
the synthesis of set functions,
we discuss their synthesis first.
To implement plural functions in Curry,
we have to decide how to represent sets
(our implementation returns multisets) of elements.
An obvious representation are lists.
Since we will later consider non-linear rules and also nested set functions
where non-determinism of different encapsulation levels are combined,
we use search trees \cite{BrasselHanusHuch04JFLP}
to represent choices between values.
The type of a search tree parameterized over the type of elements
can be defined as follows:
\begin{curry}
data ST a = Val a | Fail | Choice (ST a) (ST a)
\end{curry}
Hence, a search tree is either an expression in head-normal form,
a failure, or a choice between search trees.
Although this definition does not enforce that the argument of
a \code{Val} constructor is in head-normal form,
this invariant will be ensured by our synthesis method for set functions,
as presented below.
For instance, the plural function of the operation
\code{aBool} (\ref{ex:abool}) can be defined as
\begin{curry}
aBoolP :: ST Bool
aBoolP = Choice (Val False) (Val True)
\end{curry}
The plural function of the logical negation \code{not} defined by
\equprogram{\label{ex:not}%
not False = True\\
not True~ = False  
}
takes a search tree as an argument so that its definition must match
all search tree constructors.
Since the matching structure is similar for all operations
performing pattern matching on an argument,
we use the following generic operation to apply an operation defined
by pattern matching to a non-deterministic argument:\footnote{%
Actually, this operation is the monadic ``bind'' operation with
flipped arguments if \code{ST} is an instance of \code{MonadPlus},
as proposed in \cite{FischerKiselyovShan11}.
Here, we prefer to provide a more direct implementation.}
\begin{curry}
applyST :: (a -> ST b) -> ST a -> ST b
applyST f (Val x)        = f x
applyST _ Fail           = Fail
applyST f (Choice x1 x2) = Choice (f `applyST` x1) (f `applyST` x2)
\end{curry}
Hence, failures remain as failures, and a choice in the argument
leads to a choice in the result of the operation, which is also called
a pull-tab step \cite{AlqaddoumiAntoyFischerReck10}.
Now the plural function of \code{not} can be defined by
(shortly we will specify a systematic translation method)
\begin{curry}[mathescape=false]
notP :: ST Bool -> ST Bool
notP = applyST $ \x -> case x of False -> Val True
                                  True  -> Val False
\end{curry} 
The synthesis of plural functions for uniform programs
is straightforward: pattern matching is implemented with \code{applyST}
and function composition in right-hand sides comes for free.
For instance, the plural function of
\begin{curry}
twiceNot x = not (not x)
\end{curry}
is
\begin{curry}
twiceNotP x = notP (notP x)
\end{curry}
So far we considered only base values in search trees.
If one wants to deal with structured data, like lists of integers,
a representation like \code{ST [Int]} is not appropriate
since non-determinism can occur in any constructor of the list,
as shown by
\begin{curry}
one23 = (1?2) : ([] ? (3:[]))
\end{curry}
The expression \code{one23} evaluates to \code{[1]}, \code{[2]},
\code{[1,3]}, and \code{[2,3]}.
If we select only the head of the list, the non-determinism in the
tail does not show up, i.e., \code{head$\;$one23} evaluates
to two values \code{1} and \code{2}.
This demands for a representation of head-normal forms
with possible search tree arguments.
The head-normal forms of non-deterministic lists are the usual list
constructors where the cons arguments are search trees:
\begin{curry}
data STList a = Nil | Cons (ST a) (ST (STList a))
\end{curry}
The plural representation of \code{one23} is
\begin{curry}
one23P :: ST (STList Int)
one23P = Val (Cons (Choice (Val 1) (Val 2))
                   (Choice (Val Nil) (Val (Cons (Val 3) (Val Nil)))))
\end{curry}
The plural function of \code{head} is synthesized as
\begin{curry}[mathescape=false]
headP :: ST (STList a) -> ST a
headP = applyST $ \xs -> case xs of Nil      -> Fail
                                     Cons x _ -> x
\end{curry} 
so that \code{headP$\;$one23P} evaluates to
\code{Choice$\;$(Val$\;$1) (Val$\;$2)}, as intended.

To provide a precise definition of this transformation,
we assume that all operations in the program are uniform
(see Sect.~\ref{sec:flp}).
The plural transformation $\pt{\cdot}$ of these kinds
of function definitions is defined as follows
(where $\plural{C}$ denotes the constructor of the non-deterministic type,
like \code{STList}, corresponding to the original constructor $C$):
\begin{eqnarray*}
\pt{f~x_1 \ldots x_n~\code{=}~e} & ~=~ &
\plural{f}~x_1 \ldots x_n~\code{=}~ \pt{e}\\[2ex]
\pt{\begin{array}{@{}r@{~~}c@{~~}l@{}}
f~x_1 \ldots x_{n-1}~(C^1~x_{11} \ldots x_{1i_1})~\code{=}~e_1\\
\vdots\\
f~x_1 \ldots x_{n-1}~(C^n~x_{n1} \ldots x_{ni_n})~\code{=}~e_n
\end{array}} & ~=~ &
\begin{array}{@{}l}
f~x_1 \ldots x_{n-1}~\code{=}~
 \code{applyST~\$~\bs{}}x \to \\
 ~~~\code{case}~x~\code{of}\\
 ~~~~~~\plural{C^1}~x_{11} \ldots x_{1i_1} \to \pt{e_1}\\
 ~~~~~~\vdots\\
 ~~~~~~\plural{C^n}~x_{n1} \ldots x_{ni_n} \to \pt{e_n}
\end{array}
\end{eqnarray*}
Note that $x_i$ and $x_{jk}$ have different types in the
original and transformed program, e.g.,
an argument of type \code{Int} is transformed
into an argument of type \code{ST$\;$Int}.
Furthermore, expressions occurring in the function bodies are
transformed according to the following rules:
\begin{eqnarray*}
\pt{x} & ~=~ & x \\
\pt{C~e_1 \ldots e_n} & ~=~ & \code{Val~(}\plural{C}~ \pt{e_1} \ldots \pt{e_n}\code{)} \\
\pt{f~e_1 \ldots e_n} & ~=~ & \plural{f}~ \pt{e_1} \ldots \pt{e_n} \\
\pt{e_1 ~\code{?}~e_2} & ~=~ & \code{Choice}~ \pt{e_1}~\pt{e_2} \\
\pt{\code{failed}} & ~=~ & \code{Fail}
\end{eqnarray*}
The presented synthesis of plural functions is simple
and yields compositionality and laziness.
Thus, they are a good basis to define set functions, as shown next.

\section{Synthesis of Set Functions: The Simple Way}
\label{sec:toplevelsets}

Plural functions take sets as arguments whereas set functions are
applied to standard expressions which might not be evaluated.
To distinguish these possibly unevaluated arguments from head-normal forms,
we add a new constructor to search trees
\begin{curry}
data ST a = Val a | Uneval a | Fail | Choice (ST a) (ST a)
\end{curry}
and extend the definition of \code{applyST} with the rule
\begin{curry}
applyST f (Uneval x) = f x
\end{curry}
Furthermore, plural functions yield non-deterministic structures
which might not be completely evaluated.
By contrast, set functions yield sets of values,
i.e., completely evaluated elements.
In order to turn a plural function into a set function,
we have to evaluate the search tree structure into the
set of their values.
For the sake of simplicity, we represent the latter as ordinary lists.
Thus, we need an operation like
\begin{curry}
stValues :: ST a -> [a]
\end{curry}
to extract all the values from a search tree. For instance, the expression
\begin{curry}
stValues (Choice (Val 1) (Choice Fail (Val 2)))
\end{curry}
should evaluate to the list \code{[1,2]}.
This demands for the evaluation of all the values in a search tree
(which might be head-normal forms with choices at argument positions)
into its complete normal form.
We define a type class\footnote{Although the current definition of
Curry \cite{Hanus16Curry} does not include type classes,
many implementations of Curry,
like PAKCS, KiCS2, or MCC, support them.}
for this purpose:
\begin{curry}
class NF a where
  nf :: a -> ST a$\listline$
\end{curry}
Each instance of this type class must define a method \code{nf}
which evaluates a given head-normal form into a search tree
where all \code{Val} arguments are completely evaluated.
Instances for base types are easily defined:
\begin{curry}
instance NF Int where
  nf x = Val x
\end{curry}
The operation \code{nf} is easily extended to arbitrary search trees:\footnote{%
The use of \code{seq} ensures that the \code{Uneval} argument is evaluated.
Thus, non-determinism and failures in arguments of set functions
are not encapsulated, as intended.}
\begin{curry}
nfST :: NF a => ST a -> ST a
nfST (Val x)        = nf x
nfST (Uneval x)     = x `seq` nf x
nfST Fail           = Fail
nfST (Choice x1 x2) = Choice (nfST x1) (nfST x2)
\end{curry}
Now we can define an operation that collects all the values in a search tree
(without \code{Uneval} constructors)
into a list by a depth-first strategy:
\begin{curry}
searchDFS :: ST a -> [a]
searchDFS (Val x)        = [x]
searchDFS Fail           = []
searchDFS (Choice x1 x2) = searchDFS x1 ++ searchDFS x2
\end{curry}
Thus, failures are ignored and choices are concatenated.
Combining these two operations yields the desired definition of
\code{stValues}:
\begin{curry}
stValues :: NF a => ST a -> [a]
stValues = searchDFS . nfST
\end{curry}
\code{NF} instances for structured types can be defined by moving
choices and failures in arguments to the root:
\begin{curry}
instance NF a => NF (STList a) where
  nf Nil         = Val Nil
  nf (Cons x xs) = case nfST x of
    Choice c1 c2 -> Choice (nf (Cons c1 xs)) (nf (Cons c2 xs))
    Fail         -> Fail
    y            -> case nfST xs of
      Choice c1 c2 -> Choice (nf (Cons y c1)) (nf (Cons y c2))
      Fail         -> Fail
      ys           -> Val (Cons y ys)
\end{curry}
For instance, the non-deterministic list value \code{[1?2]}
can be described by the \code{ST} structure
\begin{curry}
nd01 = Val (Cons (Choice (Val 0) (Val 1)) (Val Nil))
\end{curry}
so that \code{stValues$\;$nd01} moves the inner choice to the top-level
and yields the list
\begin{curry}
[Cons (Val 0) (Val Nil), Cons (Val 1) (Val Nil)]
\end{curry}
which represents the set $\{\code{[0]}, \code{[1]}\}$.

As an example for our first approach to synthesize set functions,
consider the following operation (from Curry's prelude)
which non-deterministically returns any element of a list:
\begin{curry}
anyOf :: [a] -> a
anyOf (x:xs) = x ? anyOf xs
\end{curry}
Since set functions do not encapsulate non-determinism caused by arguments,
the expression \code{\setfunct{anyOf}$\;$[0?1,2,3]} evaluates to the sets
$\{\code{0},\code{2},\code{3}\}$ and
$\{\code{1},\code{2},\code{3}\}$.

In order to synthesize the set function for \code{anyOf}
by exploiting plural functions, we have to convert ordinary types,
like \code{[Int]}, into search tree types, like \code{ST (STList Int)},
and vice versa.
For this purpose, we define two conversion operations for each type
and collect their general form in the following type class:\footnote{%
Multi-parameter type classes are not yet supported
in the Curry systems PAKCS and KiCS2.
The code presented here is more elegant,
but equivalent, to the actual implementation.}
\begin{curry}
class ConvertST a b where
  toValST   :: a -> b
  fromValST :: b -> a
\end{curry}
Instances for base and list types are easily defined:
\begin{curry}
instance ConvertST Int Int where
  toValST   = id
  fromValST = id$\listline$
instance ConvertST a b => ConvertST [a] (STList b) where
  toValST []     = Nil
  toValST (x:xs) = Cons (toST x) (toST xs)$\listline$
  fromValST Nil                     = []
  fromValST (Cons (Val x) (Val xs)) = fromValST x : fromValST xs
\end{curry}
where the operation \code{toST} is like \code{toValST} but
adds a \code{Uneval} constructor:
\begin{curry}
toST :: ConvertST a b => a -> ST b
toST = Uneval . toValST
\end{curry}
The (informal) precondition of \code{fromValST} is that its argument
is already fully evaluated, e.g., by an operation like \code{stValues}.
Therefore, we define the following operation to translate an
arbitrary search tree into the list of its Curry values:
\begin{curry}
fromST :: (ConvertST a b, NF b) => ST b -> Values a
fromST = map fromValST . stValues
\end{curry}
As already mentioned, we use lists to represent multisets of values:
\begin{curry}
type Values a = [a]
\end{curry}
However, one could also use another (abstract) data type to represent
multisets or even convert them into sets, if desired.

Now we have all parts to synthesize a set function:
convert an ordinary value into its search tree representation,
apply the plural function on it, and translate the search tree
back into the multiset (list) of the values contained in this tree.
We demonstrate this by synthesizing the set function of \code{anyOf}.

The uniform representation of \code{anyOf} performs complete
pattern matching on all constructors:
\begin{curry}
anyOf []     = failed
anyOf (x:xs) = x ? anyOf xs
\end{curry}
We easily synthesize its plural function according to the scheme
of Sect.~\ref{sec:plural}:
\begin{curry}[mathescape=false]
anyOfP :: ST (STList Int) -> ST Int
anyOfP = applyST $ \xs ->
           case xs of Nil       -> Fail
                      Cons x xs -> Choice x (anyOfP xs)
\end{curry} 
Finally, we obtain its set function by converting the argument
into the search tree and the result of the plural function
into a multiset of integers:
\begin{curry}
anyOfS :: [Int] -> Values Int
anyOfS = fromST . anyOfP . toST
\end{curry}
The behavior of our synthesized set function is identical to their
original definition, e.g., \code{anyOfS$\;$[0?1,2,3]} evaluates to the
lists \code{[0,2,3]} and \code{[1,2,3]}, i.e., non-determinism
caused by arguments is not encapsulated. This is due to the fact
that the evaluation of arguments, if they are demanded inside the set function,
are initiated by standard pattern matching so that a potential
non-deterministic evaluation leads to a non-deterministic evaluation
of the synthesized set function.

In contrast to the strict evaluation of set functions in PAKCS,
as discussed in the introduction, our synthesized set functions
evaluate their arguments lazily.
For instance, the set function of \code{ndconst} defined
in Example~\ref{ex:ndconst} is synthesized as follows:
\begin{curry}
ndconstP :: ST Int -> ST Int -> ST Int
ndconstP nx ny = Choice nx (Val 1)$\listline$
ndconstS :: Int -> Int -> Values Int
ndconstS x y = fromST (ndconstP (toST x) (toST y))
\end{curry}
Since the second argument of \code{ndconstS} is never evaluated,
the expression \code{ndconstS$\;$2$\;$failed} evaluates to \code{[2,1]}
and \code{ndconstS$\;$(2?4)$\;$(3?5)}
yields the lists \code{[2,1]} and \code{[4,1]}.
The set function implementation of PAKCS fails on the first expression
and yields four results on the second one.
Hence, our synthesized set function yields better results than PAKCS,
in the sense that it is more complete and avoids duplicated results.
Morever, specific primitive operations, like \code{findall},
are not required.

The latter property is also interesting from another point of view.
Since PAKCS uses Prolog's \code{findall}, the evaluation
strategy is fixed to a depth-first search strategy implemented by
backtracking.
Our implementation allows more flexible search strategies
by modifying the implementation of \code{stValues}.
Actually, one can generalize search trees and \code{stValues}
to a monadic structure, as done in
\cite{BrasselFischerHanusReck11,FischerKiselyovShan11},
to implement various strategies for non-deterministic programming.

A weak point of our current synthesis is the handling of failures.
For instance, the evaluation of \code{anyOfS$\;$[failed,1]}
fails (due to the evaluation of the first list element)
whereas $\code{\setfunct{anyOf}}\;\code{[failed,1]} = \{\code{1}\}$
according to Def.~\ref{set-function-definition}.
To correct this incompleteness, failures resulting from argument
evaluations must be combined with result sets.
This can be done by extending search trees and distinguishing
different sources of failures, but we omit it here since
a comprehensive solution to this issue will be presented
in Sect.~\ref{sec:nested} when nested applications of set functions
are discussed.

\section{Adding Call-Time Choice}
\label{sec:nonlinear}

We have seen in Sect.~\ref{sec:introduction} that the expression
\code{double$\;$(0?1)} should evaluate to the values \code{0}
or \code{2} due to the call-time choice semantics.
Thus, the set function of
\begin{curry}
double01 :: Int
double01 = double (0?1)
\end{curry}
should yield the multiset $\{\code{0},\code{2}\}$.
However, with the current synthesis, the corresponding
set function yields the list \code{[0,1,1,2]}
and, thus, implements the run-time choice.
The problem arises from the fact that the non-determinstic
choice in the synthesized plural function
\begin{curry}
double01P :: ST Int
double01P = doubleP (Choice (Val 0) (Val 1))
\end{curry}
is duplicated by \code{doubleP}.
In order to implement the call-time choice, the same
decision (left or right choice) for both duplicates has to be made.
Instead, the search operation \code{searchDFS} handles these
choices independently and is unaware of the duplication.

To tackle this problem, we follow the idea implemented in
KiCS2 \cite{BrasselHanusPeemoellerReck11} and extend our
search tree structure by identifiers for choices
(represented by the type \code{ID}) as follows:
\begin{curry}
data ST a = Val a | Uneval a | Fail | Choice ID (ST a) (ST a)
\end{curry}
The changes to previously introduced operations on search trees,
like \code{applyST} or \code{nfST}, are minimal and straightforward
as we only have to keep a choice's identifier in their definitions.
The most significant change occurs in the search operation.
As shown in \cite{BrasselHanusPeemoellerReck11}, the call-time choice
can be implemented by storing the decision for a choice,
when it is made for the first time, during the traversal
of the search tree and looking it up later when encountering
the same choice again.
We introduce the type
\begin{curry}
data Decision = Left | Right
\end{curry}
for decisions and use an association list\footnote{Of course,
one can replace such lists by more efficient access structures.}
as an additional
argument to the search operation to store such decisions.
The adjusted depth-first search then looks as follows:
\begin{curry}
searchDFS :: [(ID,Decision)] -> ST a -> [a]
searchDFS _ (Val x)          = [x]
searchDFS _ Fail             = []
searchDFS m (Choice i x1 x2) = case lookup i m of
  Nothing    -> searchDFS ((i,Left):m) x1 ++
                  searchDFS ((i,Right):m) x2
  Just Left  -> searchDFS m x1
  Just Right -> searchDFS m x2
\end{curry}
When extracting all the values from a search tree, we initially
pass an empty list to the search operation since no decisions
have been made at that point:
\begin{curry}
stValues :: NF a => ST a -> [a]
stValues = searchDFS [] . nfST
\end{curry}
Finally, we have to ensure that the choices
occurring in synthesized plural functions are provided
with unique identifiers.
To this end, we assume a type \code{IDSupply} that
represents an infinite set of such identifiers
along with the following operations:
\begin{curry}
initSupply              :: IDSupply
uniqueID                :: IDSupply -> ID
leftSupply, rightSupply :: IDSupply -> IDSupply
\end{curry}
The operation \code{initSupply} yields an initial
identifier set. The operation \code{uniqueID} yields
an identifier from such a set while the operations
\code{leftSupply} and \code{rightSupply} both yield
disjoint subsets without the identifier obtained by
\code{uniqueID} (see \cite{BrasselHanusPeemoellerReck11} for a
discussion about implementing these operations.).
When synthesizing plural functions,
we add an additional argument of type \code{IDSupply}
and use the aforementioned operations on it to
provide unique identifiers to every choice.
The synthesized set function has to pass the initial
identifier supply \code{initSupply} to the plural function.
In the case of \code{double01}, it looks as follows:
\begin{curry}
double01P :: IDSupply -> ST Int
double01P s = doubleP (leftSupply s)
                      (Choice (uniqueID s) (Val 0) (Val 1))$\listline$
double01S :: Values Int
double01S = fromST (doubleP initSupply)
\end{curry}
With this modified synthesis, the set function yields the expected
result \code{[0,2]}.
Note that this extended scheme is necessary only if
some operation involved in the definition of the set function
has rules with non-linear right-hand sides, i.e.,
might duplicate argument expressions.
For the sake of readability, we omit this extension in
the next section where we present another extension
necessary when set functions are nested.

\section{Synthesis of Nested Set Functions}
\label{sec:nested}

So far we considered the synthesis of set functions that
occur only at the top-level of functional computations,
i.e., which are not nested inside other set functions.
The synthesis was based on the translation of functions
involved in the definition of a set function into plural functions
and extracting all the values represented by a search tree into a list
structure.
If set functions are nested, the situation becomes more complicated
since one has to define the plural function of an inner set function.
Moreover, choices and failures produced by different set functions,
i.e., levels of encapsulations, must be distinguished
according to \cite{ChristiansenHanusReckSeidel13PPDP}.
Although nested set functions are seldom used,
a complete implementation of set functions must consider them.
Therefore, we discuss in this section how we can extend
the scheme proposed so far to accommodate nested set functions.

The original proposal of set functions \cite{AntoyHanus09}
emphasized the idea to distinguish non-determinism of arguments
from non-determinism of the function definition.
However, the influence of failing computations and the combination
of nested set functions was not specified.
These aspects are discussed in \cite{ChristiansenHanusReckSeidel13PPDP}
where a denotational semantics for functional logic programs
with weak encapsulation is proposed.
Roughly speaking, an encapsulation level is attached
to failures and choices. These levels are taken into account when
value sets are extracted from a nested non-determinism structure
to ensure that failures and choices are encapsulated
by the function they belong to and not any other.
We can model this semantics by extending the structure of search
trees as follows:
\begin{curry}
data ST a = Val a | Uneval a | Fail Int | Choice Int (ST a) (ST a)
\end{curry}
The additional argument of the constructors \code{Fail} and
\code{Choice} specifies the encapsulation level.

Consider the definition
\equprogram{\label{ex:notf}%
notf = $\setfunct{\code{not}}$ failed
}
and the expression \code{\setfunct{notf}}.
Although the right-hand side of \code{notf} fails
because the argument of \code{not} is demanded
w.r.t.\ the definition (\ref{ex:not}), the source of the
failure is inside its definition so that the failure is encapsulated
and the result of \code{\setfunct{notf}} is the empty set.
However, if we define
\equprogram{\label{ex:nots}%
nots x = $\setfunct{\code{not}}$ x
}
and evaluate \code{\setfunct{nots}$\;$failed},
the computation fails since the failure comes from outside and is
not encapsulated.
These issues are discussed in \cite{ChristiansenHanusReckSeidel13PPDP}
where it has been argued that failures outside encapsulated search
should lead to a failure instead of an empty set only if there are
no other results. For instance, the expression
\code{\setfunct{anyOf}$\;$failed} has no value (since the demanded
argument is an outside failure) whereas the value
of the expression
\equprogram{\label{ex:anyoffailed}%
\setfunct{anyOf}$\;$[failed,1]
}
is the set with the single element \code{1}.
This semantics can be implemented by comparing the levels
of failures occurring in search trees (see
\cite{ChristiansenHanusReckSeidel13PPDP} for details).

With the extension of search trees introduced above,
we are well prepared to implement this semantics in Curry itself
except for one point: outside failures always lead to a failure of
the complete evaluation if their value is needed in the encapsulated
search. Thus, the evaluation of (\ref{ex:anyoffailed})
will always fail. In order to avoid this,
we have to transform such a failure into the search tree element
\code{Fail$\;$0} (where \code{0} is the ``top'' encapsulation level,
i.e., outside any set function).
For this purpose, we modify the definitions of \code{applyST} and
\code{nfST} on arguments matching the \code{Uneval} constructor
by checking whether the evaluation
of the argument to a head-normal form fails:\footnote{%
This requires a specific primitive \code{isFail} to catch
failures, which is usually supported in Curry implementations
to handle exceptions.}
\begin{curry}
applyST f (Uneval x) = if isFail x then Fail 0 else f x$\listline$
nfST (Uneval x) = if isFail x then Fail 0 else x `seq` nf x
\end{curry}
Then one can synthesize plural and set functions similarly
to the already presented scheme.
In order to set the correct encapsulation level in \code{Fail}
and \code{Choice} constructors, every function has the current
encapsulation level as an additional argument.
Finally, one also has to synthesize plural functions of set functions
if they are used inside other set functions.
For instance, the set function of \code{not} has type
\begin{curry}
notS :: Bool -> Values Bool  
\end{curry}
but the plural function of this set function must represent
the result set again as a search tree, i.e., it has the type
\begin{curry}
notSP :: Int -> ST Bool -> ST (STList Bool)
\end{curry}
(the first argument is the encapsulation level).
To evaluate the search tree structure returned by such
plural set functions, we need an operation which
behaves similarly to \code{stValues} but returns a search tree representation
of the list of values, i.e., this operation has the type
\begin{curry}
stValuesP :: NF a => Int -> ST a -> ST (STList a)
\end{curry}
Note that this operation also takes the encapsulation level
as its first argument. For instance, failures are only encapsulated
(into an empty list) if they are on the same level, i.e.,
there is the following defining rule for \code{stValuesP}:
\begin{curry}
stValuesP e (Fail n) = if n==e then Val Nil else Fail n
\end{curry}
Choices are treated in a similar way where failures in different
alternatives are merged to their maximum level according to
the semantics of \cite{ChristiansenHanusReckSeidel13PPDP}, e.g.,
\begin{curry}
stValuesP 1 (Choice 1 (Fail 0) (Fail 1))
\end{curry}
evaluates to \code{Val Nil} (representing the empty set of values).

Now we can define \code{notSP} by evaluating the result of \code{notP}
with \code{stValuesP} where the current encapsulation level is increased:
\begin{curry}
notSP e x = stValuesP (e+1) (notP (e+1) x)
\end{curry}
The plural function of \code{notf} (\ref{ex:notf}) is straightforward
(note that the level of the generated failure is the current encapsulation
level):
\begin{curry}
notfP :: Int -> ST (STList Bool)
notfP e = notSP e (Fail e)
\end{curry}
The set function of \code{notf} is synthesized as presented before
except that we additionally provide \code{1} as the initial encapsulation
level (this is also the level encapsulated by \code{fromST}):
\begin{curry}
notfS :: Values (Values Bool)
notfS = fromST (notfP 1)
\end{curry}
As we have seen, nested set functions can be synthesized
with a scheme similar to simple set functions.
In order to correctly model the semantics of
\cite{ChristiansenHanusReckSeidel13PPDP},
an encapsulation level is added to each translated operation
which is used to generate the correct \code{Fail} and \code{Choice}
constructors.
In order integrate the synthesized set functions into
standard Curry programs, arguments passed to synthesized
set functions must be checked for failures when their values
are demanded.

The extensions presented in the previous
and this section can be combined without problems.
Concrete examples for this combination and more examples
for the synthesis techniques presented in this paper are available
on-line.\footnote{\url{https://github.com/finnteegen/synthesizing-set-functions}}
In particular, there are also examples for the synthesis
of higher-order functions, which we omitted in this paper
due to the additional complexity of synthesizing the plural
functions of higher-order arguments.



\section{Related Work}
\label{sec:related}

The problems caused by integrating encapsulated search
in functional logic programs are discussed in
\cite{BrasselHanusHuch04JFLP} where the concepts of strong
and weak encapsulation are distinguished.
Weak encapsulation fits better to declarative programming
since the results do not depend on the order of evaluation.
Set functions \cite{AntoyHanus09} makes the boundaries
between different sources of non-determinism clear.
The semantical difficulties caused by nesting set functions
are discussed in \cite{ChristiansenHanusReckSeidel13PPDP}
where a denotational semantics for set functions is presented.

The implementation of backtracking and non-determinism in functional languages
has a long tradition \cite{Wadler85}.
While earlier approaches concentrated on embedding
Prolog-like constructs in functional languages
(e.g., \cite{Hinze01,SeresSpiveyHoare99}),
the implementation of demand-driven non-determinism,
which is the core of contemporary functional logic languages
\cite{AntoyEchahedHanus00JACM},
has been less explored.
A monadic implementation of the call-time choice is developed in
\cite{FischerKiselyovShan11}
which is the basis to translate a subset of Curry
to Haskell \cite{BrasselFischerHanusReck11}.
Due to performance problems with this generic approach,
KiCS2, another compiler from Curry to Haskell, is proposed in
\cite{BrasselHanusPeemoellerReck11}.
Currently, KiCS2 is the only system implementing
encapsulated search and set functions according to
\cite{ChristiansenHanusReckSeidel13PPDP},
but the detailed implementation is complex
and, thus, difficult to maintain.
This fact partially motivated the development of the approach
described in this paper.

\section{Conclusions}
\label{sec:conclusions}

We have presented a technique to synthesize the definition
of a set function of any function defined in a Curry program.
This is useful to add set functions and encapsulated search
to any Curry system so that an explicit handling of set functions
in the run-time system is not necessary.
Thanks to our method, one can add a better (i.e., less strict)
implementation of set functions to the Prolog-based Curry
implementation PAKCS or simplify the run-time system of the
Haskell-based Curry implementation KiCS2.

A disadvantage of our approach is that it increases the size
of the transformed program
due to the addition of the synthesized code.
Considering the fact that the majority of application
code is deterministic and not involved in set functions,
the increased code size is acceptable.
Nevertheless, it is an interesting topic for future work
to evaluate this for application programs
and try to find better synthesis principles (e.g., for
specific classes of operations) which produces less additional code.

Our work has the potential of both immediate and far reaching paybacks.
We offer a set-based definition of set functions simpler and more
immediate than previous ones. We offer a notion of plural function
that is original and natural. We show interesting
relationships between the two that allow us to better understand,
reason about and compute with these concepts.
The immediate consequence is an implementation
of set functions competitive with previous proposals.

A more intriguing aspect of our work is the possibility of
replacing any non-deterministic function, $f$, in a program
with its set function, which is deterministic, by enumerating
all the results of $f$.  Thus, a (often non-deterministic) functional
logic program, would become a deterministic program.
A consequence of this change
is that the techniques for the implementation
of non-determinism, such as backtracking, bubbling and pull-tabbing,
which are the output of a tremendous intellectual effort
of the last few decades, would become unnecessary.

\appendix

\section{Properties of Plural Functions}

This section contains some interesting properties of plural functions.

\begin{lemma}
The plural function of a function is unique.
\end{lemma}
\begin{proof}
  Suppose that both $f_1$ and $f_2$ are plural functions of some function
  $f$.  Let $X$ be any set such that
  $f_1\;X=Y_1$ and $f_2\;X=Y_2$, for some $Y_1$ and $Y_2$.
  We show $Y_1 \subseteq Y_2$.  For any $y \in Y_1$, by
  Def.~\ref{plural-definition}, point (1), applied to $f_1$, there
  exists some $x$ in $X$ such that $f\;x=y$.
  Since $x \in X$ and $f_2$ is a plural function of $f$,
  Def.~\ref{plural-definition}, point (2), implies that $y$ is in $Y_2$.
  By symmetry, $Y_2 \subseteq Y_1$.
  Hence,  $f_1 = f_2$.
  \hfill$\qed$
\end{proof}

\begin{lemma}
  If $f$ and $g$ are composable functions,
  then $\plural{(f \circ g)} = \plural f \circ \plural g$.
\end{lemma}
\begin{proof}
  First, we prove that for any $X$, $\plural{(f \circ g)}\;X
  \supseteq (\plural f \circ \plural g)\;X$.
  Suppose $\plural f\;(\plural g\;X) = Z$ for some sets $X$ and $Z$.
  There exists a set $Y$ such that $\plural g\;X = Y$ and $\plural
  f\;Y = Z$.  If $z$ is some element of $Z$, then there exists some
  $y$ in $Y$ such that $z$ is a value of $g\;y$, and there exists some
  $x$ in $X$ such that $y$ is a value $f\;x$. Consequently $z$ is a value
  of $(f \circ g)(x)$ and $z$ is an element of $\plural{(f \circ g)}\;X$.
  The proof that for any $X$, $\plural{(f \circ g)}\;X \subseteq
  (\plural f \circ \plural g)\;X$ is similar.
  \hfill$\qed$
\end{proof}
The following claim establishes key relationships between the set and the
plural functions of a function.
\begin{theorem}
  \label{theorem-plural}
  For any function $f$, argument $x$ of $f$, and argument $X$ of $\plural f$:
  \begin{enumerate}
  \item {}
    $\setfunct f\;x = \plural f\;\setof x$ and
  \item {}
    $\plural f\;X = \uplus\; \setfunct f\;x, \;\; \forall x \in X$.
  \end{enumerate}
\end{theorem}
\begin{proof}
  We prove that $\setfunct f\;x \subseteq \plural f\;\setof x$.
  If $y \in \setfunct f\;x$, then, by Def.~\ref{set-function-definition},
  $y$ is a value of $f\;x$, then, by Def.~\ref{plural-definition},
  $y$ is an element of $\plural f\;\setof x$.
  The proof that $\setfunct f\;x \supseteq \plural f\;\setof x$ is
  similar. Hence condition (1) holds.
  \\
  We now prove that $\plural f\;X \subseteq \uplus\; \setfunct f\;x, \; \forall x \in X$.
  For any $X$, if $y \in \plural f\;X$, by Def.~\ref{plural-definition},
  there exists some $x \in X$ such that $y$ is a value of $f\;x$.
  By Def.~\ref{set-function-definition}, $y \in \setfunct f\;x$.
  The proof that $\plural f\;X \supseteq \uplus\; \setfunct f\;x, \; \forall x \in X$ is similar. Hence condition (2) holds.
  \hfill$\qed$
\end{proof}

\end{document}